\documentclass[a4paper, USenglish, cleveref, thm-restate]{lipics-v2021}

\hideLIPIcs
\usepackage{wrapfig}

\graphicspath{{./figures/}}

\title{Optimal In-Place Compaction of Sliding Cubes}

\author{Irina Kostitsyna}{TU Eindhoven, The Netherlands}{i.kostitsyna@tue.nl}{https://orcid.org/0000-0003-0544-2257}{}
\author{Tim Ophelders}{Utrecht University, The Netherlands \and TU Eindhoven, The Netherlands}{t.a.e.ophelders@uu.nl}{https://orcid.org/0000-0002-9570-024X}{}
\author{Irene Parada}{Universitat Politècnica de Catalunya, Spain}{irene.parada@upc.edu}{https://orcid.org/0000-0003-2401-8670}{}
\author{Tom Peters}{TU Eindhoven, The Netherlands}{t.peters1@tue.nl}{https://orcid.org/0000-0002-2702-7532}{}
\author{Willem Sonke}{TU Eindhoven, The Netherlands}{w.m.sonke@tue.nl}{https://orcid.org/0000-0001-9553-7385}{}
\author{Bettina Speckmann}{TU Eindhoven, The Netherlands}{b.speckmann@tue.nl}{https://orcid.org/0000-0002-8514-7858}{}

\authorrunning{I.\ Kostitsyna et al.}

\Copyright{Irina Kostitsyna, Tim Ophelders, Irene Parada, Tom Peters, Willem Sonke, and Bettina Speckmann}

\nolinenumbers 

\newcommand{\Z}{\mathbb{Z}}
\newcommand{\N}{\mathbb{N}}
\newcommand{\C}{\mathcal{C}}
\renewcommand{\P}{\mathcal{P}}
\newcommand{\HL}{\mathcal{LH}}
\renewcommand{\L}{\mathcal{L}}
\newcommand{\U}{\mathcal{U}}

\newcommand{\shove}{\mathrm{shove}}

\newcommand{\pillar}[4]{\ensuremath{\langle #1, #2, #3 \, .. \, #4 \rangle}}

\newcommand{\enumit}[1]{\textcolor{lipicsGray}{\sffamily\bfseries\upshape\mathversion{bold}(#1)}}

\ccsdesc[100]{Theory of computation~Computational geometry}

\keywords{Sliding cubes, Reconfiguration algorithm, Modular robots}

\begin{document}

\maketitle

\begin{abstract}
    The sliding cubes model is a well-established theoretical framework that supports the analysis of reconfiguration algorithms for modular robots consisting of face-connected cubes.
    The best algorithm currently known for the reconfiguration problem, by Abel and Kominers [arXiv, 2011], uses $O(n^3)$ moves to transform any $n$-cube configuration into any other $n$-cube configuration.
    As is common in the literature, this algorithm reconfigures the input into an intermediate canonical shape.
    In this paper we present an in-place algorithm that reconfigures any $n$-cube configuration into a compact canonical shape using a number of moves proportional to the sum of coordinates of the input cubes.
    This result is asymptotically optimal.
    Furthermore, our algorithm directly extends to dimensions higher than three.
\end{abstract}

\section{Introduction}
\label{sec:introduction}

Modular robots consist of a large number of comparatively simple robotic units. These units can attach and detach to and from each other, move relative to each other, and in this way form different shapes or configurations. This shape-shifting ability allows modular robots to robustly adapt to previously unknown environments and tasks. 
One of the major questions regarding modular robots is \emph{universal reconfiguration}: is there a sequence of moves which transforms any two given configurations into each other, and if so, how many moves are necessary? There are a variety of real-world mechatronics or theoretical computational models for modular robots and the answer to the universal reconfiguration question differs substantially between systems~\cite{Akitaya21-pivotingcharacterization}. 

In this paper, we study the \emph{sliding cube model}, which is a well-established theoretical framework that supports the analysis of reconfiguration algorithms for modular robots consisting of face-connected cubes.
In this model, a module (cube) can perform two types of moves: straight-line moves called \emph{slides} and moves around a corner called \emph{convex transitions} (see Figure~\ref{fig:moves}). During a move, the configuration (excluding the moving cube) must stay connected. Furthermore, there have to be sufficient empty cells in the corresponding grid to perform the move (see again Figure~\ref{fig:moves}: the cells indicated with red wire frames must be empty for the move to be possible). 
Maintaining connectivity during a sequence of moves is the main challenge when developing algorithms in the sliding cube model. This connectivity is crucial for most actual modular robotic systems since it allows them to retain their structure, communicate, and share other resources such as energy.

\begin{figure}[t]
    \centering
    \includegraphics{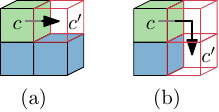}
    \caption{Moves in the sliding cube model: slide (a) and convex transition (b). Solid cubes are part of the configuration, wire frames indicate grid cells that must be empty.}
    \label{fig:moves}
\end{figure}

Almost 20 years ago, Dumitrescu and Pach~\cite{cubes-06} showed that the sliding cube model in 2D (or \emph{sliding square model}) is universally reconfigurable. 
More precisely, they presented an algorithm that transforms any two given configurations with $n$ squares into each other in $O(n^2)$ moves. 
This algorithm transforms any given configuration into a canonical horizontal line and then reverts the procedure to reach the final configuration. 
It was afterwards adapted to be \emph{in-place} using flooded bounding boxes as canonical intermediate configurations~\cite{MS20-flooding}.
Recently, Akitaya et al.~\cite{a2022compacting} presented Gather\&Compact: an input-sensitive in-place algorithm which uses $O(Pn)$ moves, 
where $P$ is the maximum among the perimeters of the bounding boxes of the initial and final configurations. 
The authors also show that minimizing the number of moves required to reconfigure is NP-hard. 

These algorithms in 2D do not directly transfer to 3D: they fundamentally rely on the fact that a connected cycle of squares encloses a well-defined part of the configuration. 
One could ask whether the fact that enclosing space in 3D is more difficult has positive or negative impact on universal reconfiguration in 3D.
Miltzow et al.~\cite{MiltzowPSSW20-3Dcubes} showed that a direct path to reconfiguration in 3D is blocked: there exist 3D configurations in which no module on the external boundary is able to move without disconnecting the configuration. Hence, simple reconfiguration strategies~\cite{hetero-03,slideonly-icalp17} can generally not guarantee reconfiguration for all instances. The best algorithm currently known for the reconfiguration problem in 3D, by 
Abel and Kominers~\cite{hypercubes-v3}, uses $O(n^3)$ moves to transform any $n$-cube configuration into any other $n$-cube configuration, which is worst-case optimal. As is common in the literature, this algorithm reconfigures the input into an intermediate canonical shape.
The question hence arises, if---just as in 2D---there is an input-sensitive reconfiguration algorithm in 3D.

\subparagraph*{Results.} In this paper we answer this question in the affirmative: We present an in-place algorithm that reconfigures any $n$-cube configuration into a compact canonical shape using a number of moves proportional to the sum of coordinates of the input cubes. This result is asymptotically optimal. Furthermore, our algorithm directly extends to hypercube reconfiguration in dimensions higher than three.

\subparagraph*{Additional related work.}
For more restricted sliding models, for example, only allowing one of the two moves in the sliding cube model, reconfiguration is not always possible. 
Michail et al.~\cite{slideonly-icalp17} explore universal reconfiguration using \emph{helpers} or \emph{seeds} (dedicated cubes that help other cubes move). They show that the problem of deciding how many seeds are needed is in PSPACE. 

Another popular model for modular robots is the pivoting cube model, in which the modules move by rotating around an edge shared with a neighboring module.
In this model the extra free-space requirements for the moves that come from pivoting instead of sliding mean that there are configurations in which no move is possible.
Akitaya et al.~\cite{Akitaya21-pivotingcharacterization} show that the reconfiguration problem in this setting is PSPACE-complete.
In contrast, adding five additional modules to the outer boundary guarantees universal reconfigurability in 2D using $O(n^2)$ moves~\cite{AkitayaADDDFKPP21-musketeers}.
Other algorithms for pivoting modules require the absence of narrow corridors in both the initial and final configurations~\cite{FeshbachS21-pivotingnoholes,SungBRR15-pivotingholes}.
A more powerful move is to allow the modules to \emph{tunnel} through the configuration.
With it, $O(n)$ parallel steps suffice to reconfigure 2D and 3D cubes~\cite{squeeze-realistic-11,squeeze-linear-09,squeeze-dist-15}.  
However, for most real-world prototype systems, tunnelling requires the use of \emph{metamodules}~\cite{Parada21-newmetamodule} which are sets of modules which act as a single unit with enhanced capabilities, increasing the granularity of the configurations.

We require that the configuration stays connected at all times. In a slightly different model that relaxes the connectivity requirement (referred to as the \emph{backbone property}),
Fekete et al.~\cite{FeketeKKRS-socg21} show that scaled configurations of labeled squares can be efficiently reconfigured using parallel coordinated moves with a schedule that is a constant factor away from optimal.

\section{Preliminaries}

\begin{wrapfigure}[6]{r}{2.5cm}
    \raggedleft
    \vspace{-0.85\intextsep}
    \includegraphics{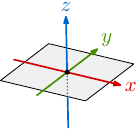}
\end{wrapfigure}
In this paper, we study cubical robots moving in the 3-dimensional grid. The handedness of the coordinate system does not have any impact on the correctness of our algorithm; in our figures we use a right-handed coordinate system with $z$ pointing up.

    A \emph{configuration}~$\C$ is a subset of coordinates in the grid.
    The elements of $\C$ are called \emph{cubes}.
    We call two cubes \emph{adjacent} if they lie at unit distance.
    For a configuration~$\C$, denote by $G_\C$ the graph with vertex set $\C$, whose edges connect all adjacent cubes.
    We say a \emph{cell} is a vertex of $G_{\Z^3}$, which is not occupied by a cube in $\C$.
    We always require a configuration to remain \emph{connected}, that is, $G_\C$ must be connected. 
    For ease of exposition we assume $\C$ consists of at least two cubes.
    We call a configuration~$\C$ \emph{nonnegative} if $\C\subseteq\N^3$.

    In the sliding cubes model, a configuration can rearrange itself by letting cubes perform moves.
    A move is an operation that replaces a single cube $c\in\C$ by another cube $c'\notin\C$.
    Moves come in two types: \emph{slides} and \emph{convex transitions}, see \cref{fig:moves}.
    In both cases, we consider a 4-cycle $\gamma$ in $G_{\Z^3}$.
    For slides, three cubes of $\gamma$ are in $\C$; $c'$ is the cube of $\gamma$ not in~$\C$, and $c$ is adjacent to~$c'$.
    For convex transitions, $\gamma$ has exactly two adjacent cubes in $\C$; $c$ is one of these two cubes, and $c'$ is the vertex of $\gamma$ not adjacent to $c$.
    The slide or convex transition is a \emph{move} if and only if $\C\setminus\{c\}$ is connected.
    
    
    Let $\C$ be a nonnegative configuration.
    Call a cube $c = (x, y, z)$ \emph{finished} if the cuboid spanned by the origin and~$c$ is completely in~$\C$, that is, if $\{0, \ldots, x\}\times \{0, \ldots, y\} \times \{0, \ldots, z\} \subseteq \C$. We call $\C$ \emph{finished} if all cubes in~$\C$ are finished.
    The \emph{compaction problem} starts with an arbitrary connected configuration $\C$ and finishes when all cubes are finished.

    Most of the algorithm works on vertical contiguous strips of cubes in~$\C$ called subpillars. More precisely, a \emph{subpillar} is a subset of~$\C$ of the form $\{x\} \times \{y\} \times \{z_b, \ldots, z_t\}$. In the remainder of this paper, we denote this subpillar by $\pillar{x}{y}{z_b}{z_t}$.
    The cube $(x, y, z_t)$ is called the \emph{head}, and the remainder $\pillar{x}{y}{z_b}{z_t - 1}$ is called the \emph{support}.
    A \emph{pillar} is a maximal subpillar, that is, a subpillar that is not contained in any other subpillar. Note that there can be multiple pillars with the same $x$- and $y$- coordinate above each other, as long as there is a gap between them.
    Two sets $S$ and~$S'$ of cubes are \emph{adjacent} if $S$ contains a cube adjacent to a cube in~$S'$.
    The \emph{pillar graph} $\P_S$ of a set~$S$ of cubes is the graph whose vertices are the pillars of~$S$ and whose edges connect adjacent pillars.

\section{Algorithm}\label{sec:algo}

    For a set of cubes $S \subseteq \C$, denote its \emph{coordinate vector sum} by $(X_S, Y_S, Z_S) = \sum_{(x, y, z)\in S} (x, y, z)$.
    Let $\C_{>0}$ be the subset of cubes $(x, y, z)\in\C$ for which $z > 0$, and $\C_0$ be the subset of cubes for which $z = 0$.
    Let the \emph{potential} of a cube $c = (c_x, c_y, c_z)$ be $\Pi_c = w_c(c_x + 2c_y + 4c_z)$, where the \emph{weight} $w_c$ depends on the coordinates of $c$ in the following way.
    If $c_z > 1$, then $w_c = 5$, if $c_z = 1$, then $w_c = 4$.
    If $c_z = 0$, then $w_c$ depends on $c_y$.
    If $c_y > 1$, then $w_c = 3$, if $c_y = 1$, then $w_c = 2$ and lastly, if both $c_z = c_y = 0$, then $w_c = 1$.
    We aim to minimize the \emph{potential function} $\Pi_\C = \sum_{c\in\C} \Pi_c$.
    From now on, let $\C$ be an unfinished nonnegative configuration.
    We call a sequence of $m$ moves \emph{safe} if the result is a nonnegative instance $\C'$, such that $\Pi_{\C'} < \Pi_{\C}$ and $m = O(\Pi_{\C} - \Pi_{\C'})$.
    This means that the sequence of moves reduces the potential by at least $m$ by going from $\C$ to $\C'$.
    We show that if $\C$ is unfinished, it always admits a safe move sequence.

    The main idea is as follows.
    For a configuration~$\C$, whenever possible, we try to reduce $Z_\C$.
    If that is not possible, then the configuration must admit another type of move, where a complete pillar is moved to a different $x$- and $y$-coordinate.
    In this way, by reducing either the $z$-coordinate of cubes, or the $x$- or $y$-coordinate, we guarantee that eventually every cube becomes finished.

\subparagraph*{Local Z reduction.}
    Let $P = \pillar{x}{y}{z_b}{z_t}$ be a subpillar of $\C$.
    We refer to the four coordinates $\{(x - 1, y), (x + 1, y), (x, y - 1), (x, y + 1)\}$ as the \emph{sides} of~$P$.
    On each side, $P$ may have one or more adjacent pillars. We order these by their $z$-coordinates; as such, we may refer to the top- or bottommost adjacent pillar on a side of~$P$.
    We say that a set of cubes $S\subseteq\C$ is \emph{non-cut} if $G_{\C\setminus S}$ is connected or empty.
    A pillar of~$\C$ is non-cut if and only if it is a non-cut vertex of the pillar graph $\P_\C$.
    
    Let $P = \pillar{x}{y}{z_b}{z_t}$ be a non-cut subpillar, and let $P' = \pillar{x'}{y'}{z'_b}{z'_t}$ be a pillar adjacent to~$P$. We define a set of operations of at most three moves within~$P$ which locally reduce $Z_\C$ (see \Cref{fig:pillar_gaps}).
    Because $P$ is non-cut, $\C \setminus P$ is connected. Therefore, if for each of these operations cubes of~$P$ move in such a way that each component (of cubes originating from~$P$) remains adjacent to a cube of $\C\setminus P$, then the result is a valid configuration.
    
    \begin{figure}[b]
        \centering
        \includegraphics{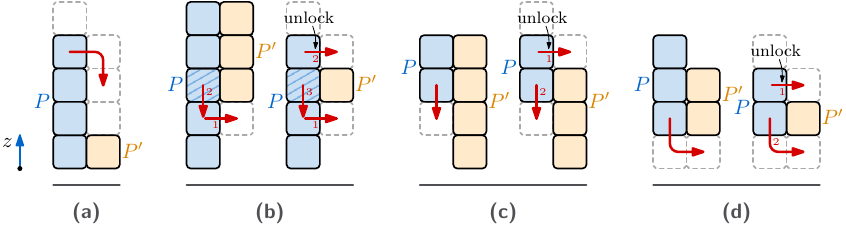}
        \caption{Examples of operations \enumit{\ref{enum:pillar_a}--\ref{enum:pillar_e}}; hatched cubes are non-cut and dashed outlines indicate cells that must be empty. Each case admits a move sequence that reduces $Z_{\C}$.}
        \label{fig:pillar_gaps}
    \end{figure}
    
    \begin{alphaenumerate}
        \item \label{enum:pillar_a}
        If $P'$ is a topmost adjacent pillar of~$P$ and $z'_t \leq z_t - 2$, then the topmost cube of $P$ admits a convex transition that decreases $Z_{\C}$.
        
        
        
        
        \item \label{enum:pillar_c}
        If $z'_b > z_b$ and $(x, y, z'_b - 1)$ is not a cut cube (i.e., $\pillar{x}{y}{z_b}{z'_b - 2}$ has an adjacent cube in $\C$, or $z'_b = z_b + 1$), then there is a move sequence that decreases $Z_{\C}$: first slide $(x, y, z'_b - 1)$ to $(x', y', z'_b - 1)$ and then slide $(x, y, z'_b)$ to $(x, y, z'_b - 1)$.
        There is one special case. We say that $P$ is \emph{locked} if the head of $P$ has no adjacent cubes except for $P$'s support. If $P$ is locked and $z'_b = z_t - 1$, then the second slide would disconnect $P$'s head from the rest of the configuration. To avoid this, before performing the second slide, we \emph{unlock} $P$ by sliding the head of~$P$ from $(x, y, z_t)$ to $(x', y', z_t)$, as shown in the right part of \Cref{fig:pillar_gaps}b.
        
        
        
        \item \label{enum:pillar_d}
        If $(x, y, z_b - 1) \notin \C$ and $z'_b<z_b$, then (after unlocking $P$, if necessary) $(x, y, z_b)$ admits a slide to $(x, y, z_b - 1)$ that decreases $Z_{\C}$.
        
        
        \item \label{enum:pillar_e}
        If $(x, y, z_b - 1) \notin \C$, $P'$ is a bottommost adjacent pillar of~$P$, and $z_b = z'_b > 0$, then (after unlocking $P$, if necessary) $(x, y, z_b)$ admits a convex transition to $(x', y', z'_b - 1)$ that decreases $Z_{\C}$.
              
    \end{alphaenumerate}

    \begin{lemma}\label{lem:moves-a-d}
        Let $P = \pillar{x}{y}{z_b}{z_t}$ be a non-cut pillar. If \enumit{a--d} do not apply to~$P$, then on each side, $P$ has at most one adjacent pillar $P' = \pillar{x'}{y'}{z'_b}{z'_t}$. For these pillars~$P'$, we have $z_t \leq z'_t + 1$, and either $z_b < z'_b$ or $z_b = z'_b = 0$.
    \end{lemma}
    \begin{proof}
        Consider one side~$s$ of~$P$. Because \enumit{a} does not apply to~$P$, for any adjacent pillar $P' = \pillar{x'}{y'}{z'_b}{z'_t}$, we know that $z_t \leq z'_t + 1$. Assume that $(x, y, z'_b - 1)$ is not a cut cube. Because \enumit{b} does not apply, $z_b \geq z'_b$, and because \enumit{c} does not apply either, $z_b = z'_b$, and finally because \enumit{d} does not apply, we have $z_b = z'_b = 0$. Now assume that $(x, y, z'_b - 1)$ is a cut cube. Then because \enumit{c} does not apply, we have $z_b \leq z'_b$, and finally because \enumit{d} does not apply, we have $z_b < z'_b$ or $z_b = z'_b = 0$.
        If $z_b = z'_b = 0$ for each adjacent pillar $P'$, then each side of $P$ can have at most one pillar.
        If $z_b < z'_b$, then, because \enumit{b} does not apply, $(x, y, z'_b - 1)$ is a cut cube. Therefore, also in this case, there can be no cube adjacent to the subpillar \pillar{x}{y}{z_b}{z'_b - 2}, and therefore also in this case there can be at most one adjacent pillar on each side.
    \end{proof}
    
    
\subparagraph*{Pillar shoves.}
    Next, we consider longer move sequences that still involve a single subpillar.
    A central operation of our algorithm is a \emph{pillar shove}, which takes as parameters a subpillar $P = \pillar{x}{y}{z_b}{z_t}$ and a side $(x', y')$ of~$P$.
    The result of the pillar shove is the set of cubes
    \[
        \shove(\C, P, (x', y')) \ := \ (\C \setminus P) \ \cup \ \pillar{x'}{y'}{z_b}{z_t - 1} \ \cup \ \{(x, y, z_b)\}\text{,}
    \]
    in which the support is effectively shifted to the side $(x', y')$, and the head is effectively moved from $(x, y, z_t)$ to $(x, y, z_b)$.
    Although $\shove(\C, \pillar{x}{y}{z_b}{z_t}, (x', y'))$ is well-defined, it is not necessarily a connected configuration, let alone safely reachable from~$\C$.
    
    \begin{figure}
        \centering
        \includegraphics{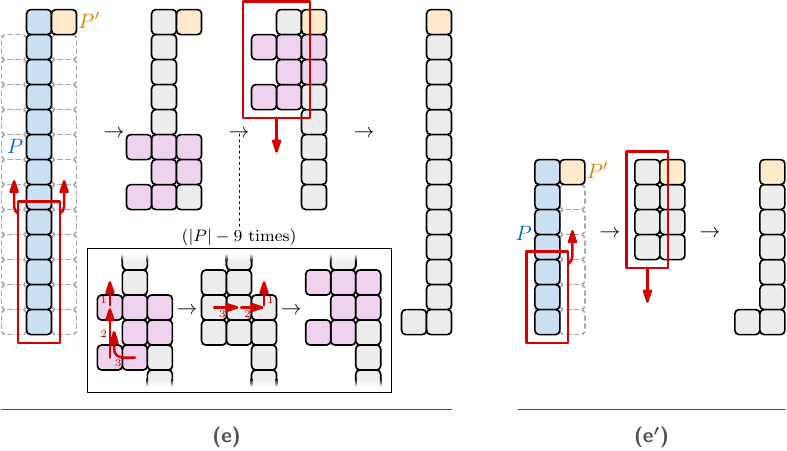}
        \caption{Example of pillar shoves for a long pillar \enumit{\ref{enum:pillar_f}} and a short pillar \enumit{\ref{enum:pillar_f}$'$}.}
        \label{fig:pillar_shove}
    \end{figure}
    
    Let $P = \pillar{x}{y}{z_b}{z_t}$ be a non-cut subpillar, and assume that on at least two sides $(x', y')$ and $(x'', y'')$ of~$P$, no cube except possibly the head $(x, y, z_t)$ has an adjacent cube.
    Moreover, assume that $(x', y', z_t) \in \C$.
    Then the pillar shove can be done without without collisions and while keeping connectivity, as illustrated in \cref{fig:pillar_shove}. There are two cases: one where $P$ has at least $9$~cubes, in which case we take $O(|P|)$ moves (illustrated on the left side of \cref{fig:pillar_shove}). The case where $P$ has fewer than~$9$ cubes takes $O(1)$ moves and does not require the existence of the second side $(x'', y'')$ (illustrated on the right side of \cref{fig:pillar_shove}).
    A pillar shove reduces $Z_{\C}$ by $z_t-z_b$ and takes $O(z_t-z_b)$ moves, so it is safe.
    
    \begin{alphaenumerate}\setcounter{enumi}{4}
        \item
        \label{enum:pillar_f}
        Assume that $P = \pillar{x}{y}{z_b}{z_t}$ has at most one adjacent pillar on each side, and there exists an adjacent pillar $P' = \pillar{x'}{y'}{z'_b}{z'_t}$ with $z'_b > z_b$ (assume that $P'$ is such a pillar with the lowest $z'_b$), and there is a side $(x'', y'')\neq (x', y')$ of~$P$ such that $(x'', y'', z_b) \notin \C$.
        Together, these assumptions imply that both sides $(x'', y'')$ and $(x', y')$ are empty up to at least $z'_b - 1$.
        Then the subpillar $\pillar{x}{y}{z_b}{z'_b}$ (after unlocking $P$, if necessary) admits a pillar shove.
    \end{alphaenumerate}

    \begin{lemma}
        Let $P = \pillar{x}{y}{z_b}{z_t}$ be a non-cut subpillar. If \enumit{\ref{enum:pillar_a}--\ref{enum:pillar_f}} do not apply to $P$, then $P$ has no adjacent pillar $\pillar{x'}{y'}{z'_b}{z'_t}$ with $z'_b > z_b$.
    \end{lemma}
    \begin{proof}
        Assume that $P$ has at least two adjacent pillars, say $P' = \pillar{x'}{y'}{z'_b}{z'_t}$ and $P'' = \pillar{x''}{y''}{z''_b}{z''_t}$, such that $z_b < z'_b$ and $z_b < z''_b$; let $P'$ denote the lowest one, such that $z_b < z'_b \leq z''_b$. Then $(x'', y'', z_b) \notin \C$, which contradicts that \enumit{\ref{enum:pillar_f}} does not apply.
        Therefore, there can be at most one such pillar.
        However, this, together with Lemma~\ref{lem:moves-a-d}, implies that on all other sides $(x'', y'')\neq(x',y')$, $(x'', y'', z_b) \in \C$.
        This means that $(x, y, z'_b-1)$ is not a cut-cube, contradicting that \enumit{\ref{enum:pillar_c}} does not apply.
        Therefore, there can be no such adjacent pillars.
    \end{proof}
    
    \noindent
    In summary, if \enumit{\ref{enum:pillar_a}--\ref{enum:pillar_f}} do not apply to any non-cut subpillar, then for any non-cut pillar $P = \pillar{x}{y}{z_b}{z_t}$ and any adjacent pillar $\pillar{x'}{y'}{z'_b}{z'_t}$, we have $z'_b=z_b=0$.
    
\subparagraph*{Local potential reduction.}
    We may greedily reduce the potential by moving individual cubes with $z > 0$.
    
    \begin{alphaenumerate}\setcounter{enumi}{5}
        \item Perform any move of $\C$ that moves a cube of $\C_{>0}$, reduces the potential, and results in a nonnegative instance.\label{enum:pillar_g}
    \end{alphaenumerate}
    
    \begin{lemma}\label{lem:single_pillar_component}
        If an unfinished configuration~$\C$ does not admit any move of type \enumit{\ref{enum:pillar_g}}, and $\pillar{x}{y}{z_b}{z_t}$ is a component of $\C_{>0}$ that consists of a single pillar, then $P=\{(0,0,1)\}$ and $(0,0,0)\in\C$.
    \end{lemma}
    \begin{proof}
        Because $\C_{>0}$ does not contain cubes with $z=0$, we have $z_b=1$ or $z_b>1$.
        If $z_b>1$, then $\C$ would be not connected, so this cannot happen.
        Likewise, if $z_b=1$ and $(x, y, 0)\notin\C$ makes $\C$ disconnected, so this cannot happen either.
        Therefore $z_b=1$ and $(x,y,0)\in\C$.
        If $z_t>1$, then the topmost cube of $P$ can do a convex transition to $(x+1,y,z_t-1)$, reducing the potential.
        Therefore $z_t=1$ and $P=\{(x,y,1)\}$.
        If $x>0$ or $y>0$, then we can move the single cube of~$P$: using the cube at $(x,y,0)$, the cube of $P$ can slide or convex transition closer to the origin $(0,0,0)$.
        Therefore, $z_b=z_t=1$ and $x=y=0$.
    \end{proof}
    \begin{corollary}
        If a nontrivial instance $\C$ does not admit any move of type \enumit{\ref{enum:pillar_g}} then $\C_{>0}$ has at most one component that consists of a single pillar.
    \end{corollary}
    
\subparagraph*{Low and high components.}
    Suppose that \enumit{\ref{enum:pillar_a}--\ref{enum:pillar_g}} do not apply.
    Let $\HL_\C$ be the bipartite graph obtained from $G_\C$ by contracting the components of $G_{\C_0}$ and $G_{\C_{>0}}$ to a single vertex (see \Cref{fig:high-low-components}).
    We call $\HL_\C$ the \emph{low-high graph} of $\C$, and we call the vertices of $\HL_\C$ that correspond to components of $G_{\C_0}$ and $G_{\C_{>0}}$ \emph{low} and \emph{high} components, respectively.

    \begin{figure}
        \centering
        \includegraphics{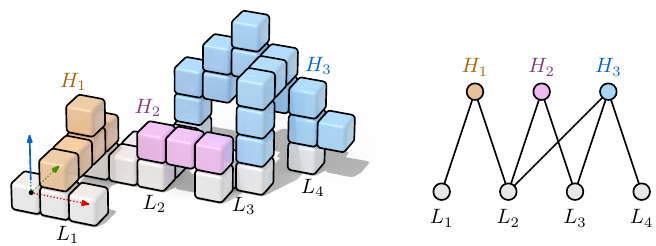}
        \caption{An example configuration~$\C$ and its low-high graph $\HL_\C$.}
        \label{fig:high-low-components}
    \end{figure}
    
    We will use the following lemma several times.
    \begin{lemma}\label{lem:common_non_cut}
        Let $H$ be a high component and $P$ be a pillar of $H$.
        For every component $H'$ of $H\setminus P$, there exists a non-cut pillar of $H'$ that is also a non-cut pillar of $H$.
    \end{lemma}
    \begin{proof}
        Any component with at least two pillars contains at least two non-cut pillars and every component contains at least one non-cut pillar, so let $P'$ be an arbitrary non-cut pillar of $H'$.
        $H\setminus P'$ has at most two components, namely $H'\setminus P'$ and $H\setminus H'$.
        If $P'$ is a non-cut pillar of $H$, the lemma holds.
        Else, if $P'$ is a cut pillar of $H$, then it has exactly these components, so $H'\setminus P'$ is nonempty and $P$ is adjacent to $P'$.
        Therefore, $H'$ consists of multiple pillars and hence at least two non-cut pillars.
        Let $P''\neq P'$ be a second non-cut pillar of $H'$.
        We claim that $P''$ is also a non-cut pillar of $H$.
        Indeed, because $P$ and $P'$ are adjacent, the sets $H'\setminus P''\supseteq P'$ and $H\setminus H'\supseteq P$ are adjacent.
        Hence, $H\setminus P''=(H'\setminus P'')\cup(H\setminus H')$ has a single component, so $P''$ is a non-cut pillar of~$H$.
    \end{proof}
    
    \begin{lemma}
        Assume $\C$ does not admit any move of type \enumit{\ref{enum:pillar_a}--\ref{enum:pillar_g}}.
        Suppose that $H$ is a high component such that $\C\setminus H$ is connected.
        Then any pillar of $H$ is a non-cut subpillar of $\C$.
    \end{lemma}
    \begin{proof}
        Suppose for a contradiction that a pillar $P$ of $H$ is a cut subpillar of $\C$.
        Then $\C\setminus P$ contains at least one component $H'$ that does not intersect $\C\setminus H$.
        Therefore, $H'$ is also a component of $H\setminus P$, so by Lemma~\ref{lem:common_non_cut}, there exists a non-cut pillar $P' = \pillar{x'}{y'}{z'_b}{z'_t}$ of $H'$ that is also a non-cut pillar of $H$.
        If $z'_b>1$ or $(x',y',0)\notin\C$, then $P'$ would be a non-cut pillar of $\C$. If $\C$ does not admit any move of type \enumit{\ref{enum:pillar_a}--\ref{enum:pillar_g}}, all non-cut pillars of $\C$ start at $z=0$, which contradicts $z'_b>1$ or $(x',y',0)\notin\C$.
        Therefore, $z'_b=1$ and $(x',y',0)\in\C$, but then $H'$ would not be a component of $\C\setminus P$, as $H'$ is adjacent to $(x',y',0)\in\C\setminus P$.
        Hence, $H'$ cannot exist, completing the proof.
    \end{proof}
    \begin{corollary}\label{col:z=0}
        If $H$ is a high component such that $\C\setminus H$ is connected, then every pillar of $H$ is part of a pillar of $\C$ starting at $z=0$.
    \end{corollary}
    
    \begin{lemma}
        Assume $\C$ does not admit any move of type \enumit{\ref{enum:pillar_a}--\ref{enum:pillar_g}}.
        If $H$ is a high component such that $\C\setminus H$ is connected, then $H$ consists entirely of finished cubes.
    \end{lemma}
    \begin{proof}
        Assume for contradiction that $H$ contains an unfinished cube.
        Because of Corollary~\ref{col:z=0}, every pillar of $H$ is part of a pillar of $\C$ starting at $z=0$.
        Therefore, $H$ contains an unfinished cube $(x,y,z)$ 
        with $x>0$ or $y>0$.
        Let $c$ be such a cube that lexicographically maximizes $(z,-y,-x)$.
        If $x>0$ and $(x-1,y,z)\notin H$ (and thus $(x-1,y,z)\notin \C$), then we can move $c$ to either $(x-1,y,z)$ or $(x-1,y,z-1)$, reducing the potential while maintaining a nonnegative instance, so \enumit{\ref{enum:pillar_g}} would apply.
        If $y>0$ and $(x,y-1,z)\notin H$, then we can similarly move $c$ to either $(x,y-1,z)$ or $(x,y-1,z-1)$.
        On the other hand, if both (1) $x=0$ or $(x-1,y,z)\in H$ and (2) $y=0$ or $(x-1,y,z)\in H$, then because $c$ is the unfinished cube of $H$ that maximizes $(z,-y,-x)$, the cubes $(x-1,y,z)$ (if $x>0$) and $(x,y-1,z)$ (if $y>0$) are finished, but then $(x,y,z)$ would also be finished. Contradiction.
    \end{proof}
    \begin{corollary}
        Any high component that contains a finished cube also contains~$(0,0,1)$. Hence there is at most one such high component.
    \end{corollary}
    
\subparagraph*{Handling low components.}
    We will now pick a vertex $R$ of $\HL_\C$ that we call the \emph{root} of $\HL_\C$.
    If $(0,0,0)\in\C$, pick $R$ to be the low component that contains $(0,0,0)$.
    Otherwise, pick $R$ to be an arbitrary low component.
    Let $d$ be the maximum distance in $\HL_\C$ from $R$ over all vertices of $\HL_\C$.
    Let $\U$ be the set of vertices of $\HL_\C$ that are locally furthest away (in $\HL_\C$) from $R$.
    That is, all neighbors $v$ of a vertex $u\in \U$ lie closer to $R$.
    All vertices of $\U$ are non-cut subsets of $\C$.
    Therefore, if $\U$ contains a high component $H$, then $H$ consists entirely of finished cubes (and $H$ is adjacent to $R$), so $\U$ contains at most one high component.

    If $d=0$, then $\C$ consists of a single low component.
    If $d=1$, then $\C$ consists of one high and one low component. Set $\U$ contains exactly one high component, and it consists entirely of finished cubes.
    We will handle these two cases later.
    If $d\geq 2$, then $\C$ consists of at least two low components and $\U$ consists of at least one low component, and at most one high component.
    
    We call a low component $L$ \emph{clear} if $\C\setminus L$ is connected, $L\neq R$, and $L$ is connected to  a non-cut pillar $P$ in $\C\setminus L$.
    We call such a pillar $P$ a \emph{clearing} pillar.
    We show in Lemma~\ref{lem:clear_low_component} that if $d\geq 2$, there is at least one clear low component.
    For this, consider a low component that is furthest from $R$ (i.e. at distance $d$), and let $H$ be an adjacent high component.
    Let $\L_H$ be the set of low components in $\U$ that are adjacent to $H$ (and hence also lie at distance $d$ from $R$).
    \begin{lemma}\label{lem:clear_low_component}
        At least one low component $L\in \L_H$ is connected to $H$ via a non-cut pillar of $\C\setminus L$.
    \end{lemma}
    \begin{proof}
        Let $\C'=\C\setminus \bigcup_{L\in \L_H} L$.
        Let $H_{\L_H}$ be the set of cubes of $H$ that are adjacent to a low component in $\L_H$.
        Fix an arbitrary cube $c_s$ of $R$, and in the graph $G_{\C'}$, consider a cube $c$ of $H_{\L_H}$ that is farthest from $c_s$.
        Let $P$ be the pillar of $H$ that contains $c$.
        We will show that $P$ is a non-cut pillar of $\C'$.
        Suppose for a contradiction that $P$ is a cut pillar, then $\C'\setminus P$ contains at least two component, only one of which does not contain $R$.
        Let $H'$ be such a component of $\C'\setminus P$ not containing $R$.
        If $H'$ contains a cube not in $H$, then that cube lies in a low or high component of $\C'$ that lies closer to $R$, which therefore remains connected to $R$ after removing $P$.
        Therefore, $H'$ is a subset of $H$.
        All cubes $(x,y,z)\in H'$ with $z=1$ lie farther from $c_s$ than $c$, and therefore $H'$ does not contain any cubes of $H_{\L_H}$, so $H'$ is not adjacent to any cubes of $\L_H$.
        Therefore, $H'$ is not adjacent to any cubes of $(\C\setminus P)\setminus H'$.
        By Lemma~\ref{lem:common_non_cut}, $H'$ contains at least one non-cut pillar $P'$ that is also a non-cut pillar of $H$.
        $P'$ is also a non-cut pillar of $\C$, but all non-cut pillars of $\C$ start at $z=0$ (Corollary~\ref{col:z=0}), which is a contradiction.
    \end{proof}
    
    We will now present an algorithm that repeatedly selects a (non-root) clear low component~$L$, and performs the following move on it:
    \begin{alphaenumerate}\setcounter{enumi}{6}
        \item Perform any move of $\C$ that moves a cube of $L$, reduces the potential, and results in a nonnegative instance. Same as \enumit{\ref{enum:pillar_g}}, but now on a low component.\label{enum:pillar_h}
    \end{alphaenumerate}

    Moves of type \enumit{\ref{enum:pillar_h}} serve any of three purposes:
    \begin{enumerate}[(1)]
        \item The component connects to a different low component, merging them.\label{enum:low_merge}
        \item The component connects to the root, and becomes part of the root.\label{enum:low_merge_root}
        \item The component reaches the origin (at which point it becomes the root).\label{enum:low_origin}
    \end{enumerate}
    In between these operations, we exhaustively perform potential reductions to high components using move types \enumit{\ref{enum:pillar_a}--\ref{enum:pillar_g}}.
    If none of the moves \enumit{\ref{enum:pillar_a}--\ref{enum:pillar_h}} are available, then this means that $L$ does not have enough cubes to reach the origin.

    In this case, we want to move the clearing pillar and do a pillar shove.
    However, it could be that there are cubes around the clearing pillar, or that the clearing pillar would disconnect the low component.
    Let $N$ be the set of grid cells~$c$ neighboring $P$ with $z=1$, such that $c$ lies within the bounding box of $\C$.
    
    \begin{lemma}\label{lem:pillar_surrounded}
        Let $\C$ be a configuration that does not admit moves of type \enumit{\ref{enum:pillar_a}--\ref{enum:pillar_h}}. Let $L$ be a clear component and $P = \pillar{x}{y}{z_b}{z_t}$ be its clearing pillar. The cubes in $N$ need to be either all present, or all absent from $\C$.
    \end{lemma}
    \begin{proof}
        Assume for contradiction that at least one, but at most three of the cubes neighboring $P$ with $z=1$ are present.
        Let these cubes be $c_1$, $c_2$, and $c_3$, not all need to be present.
        Since $P$ is a clearing pillar, $(\C\setminus L)\setminus P$ is connected.
        Therefore, the cube $(x, y, 2)$ can do a potential reducing move of type \enumit{\ref{enum:pillar_g}}, which is a contradiction.
    \end{proof}

    \begin{alphaenumerate}\setcounter{enumi}{7}
        \item Let $L$ be a clear component and $P = \pillar{x}{y}{z_b}{z_t}$ be its clearing pillar. Assume that all cubes in $N$ be present. By definition, $L\neq R$. Therefore, there exists an empty spot $e$, with coordinates $(e_x, e_y, 0)$, with $e_x < x$ or $e_y < y$. Let $e$ be such an empty spot with highest $y$, and from those, the one with highest $x$.
        We now take the cube $(x, y, 0)$ and move it on a shortest path via $z=-1$ towards $e$, reducing its potential.
        Because all cubes in $N$ are present, and because those cubes are part of $\C\setminus L$, of which $P$ is a non-cut pillar, the configuration stays connected during this move.\label{enum:pillar_i}
        \item[\labelitemi] If \enumit{\ref{enum:pillar_a}--\ref{enum:pillar_i}} do not apply, then all cubes from $N$ are not in $\C$ by Lemma~\ref{lem:pillar_surrounded}.
        \item Let $L$ be a clear component and $P = \pillar{x}{y}{z_b}{z_t}$ be its clearing pillar. All cubes in $N$ are absent from $\C$. Gather cubes from $L$ towards $P$ according to Figure~\ref{fig:pillar_shove_with_extra_cubes} and do a pillar shove on $P$. Then, move the extra cubes back to their original location.\label{enum:pillar_j}
    \end{alphaenumerate}

    The only reason that \enumit{\ref{enum:pillar_f}} is not possible, is because $P$ is a cut pillar, since it is the only pillar connecting $L$ to the other components.
    Therefore, gathering cubes from $L$ to $P$ makes this pillar shove viable.

    This is done in the following way.
    Let the clearing pillar of $L$ be $P = \pillar{x}{y}{z_b}{z_t}$.
    Note that because \enumit{\ref{enum:pillar_a}--\ref{enum:pillar_i}} do not apply, $z_b = 1$. Let $z_t$ be the highest $z$ such that only $(x, y, z_t)$ has a horizontal neighboring cube $(x', y', z_t)$. Let $(x', y')$ be the side of $P$ that is lexicographically smallest. This is the direction in which we would like to do a pillar shove. Let $c=(x, y, 0) \in L$ be the cube below $P$.
    Assume that $P$ has at least size $4$.
    Then, repeatedly select the non-cut cube~$c\in L$ that lexicographically maximizes $(z, y, x)$ and move it towards $P$ to create the configuration visible in \cref{fig:pillar_shove_with_extra_cubes}.
    Now $P$ can execute a pillar shove in the direction of $(x',y')$.
    If $P$ has less than $4$ cubes we gather cubes towards $P$ as in \cref{fig:pillar_shove_with_extra_cubes_small}.
    Then, using a constant number of moves, we can decrease the potential vector, while maintaining connectivity.

    \begin{figure}[h]
        \centering
        \includegraphics{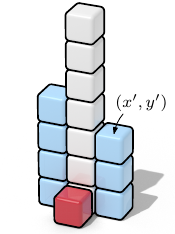}
        \caption{The start configuration for a pillar shove for a clearing pillar of height at least $4$. The blue cubes are required and need to be gathered. The red cube is part of $L$.}
        \label{fig:pillar_shove_with_extra_cubes}
    \end{figure}
    \begin{figure}[h]
        \centering
        \includegraphics{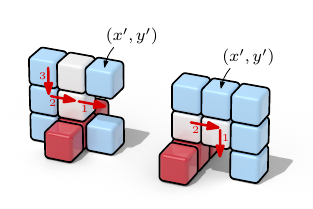}
        \caption{The start configuration for a pillar shove for a small (size less than 4) clearing pillar. The blue cubes are required to keep $L$ connected to the clearing pillar and will need to be gathered if they are not present yet. The red cubes can be part of $L$. Left: Blue cubes necessary for a clearing pillar of height $2$ or $3$. Right: configuration necessary for a clearing pillar of height $1$.}
        \label{fig:pillar_shove_with_extra_cubes_small}
    \end{figure}

    This algorithm terminates when no clear low component (and hence only the root low component) remains.
    We are left with two cases: Either no high component remains ($d = 0$), or there is at most one high component ($d=1$), which consists of entirely finished cubes.

    All of the moves \enumit{\ref{enum:pillar_a}--\ref{enum:pillar_j}} not only work in $3$ dimensions, they also work in $2$ dimensions when instead of prioritizing reducing the $z$-coordinate, we prioritize reducing the $y$-coordinate.
    Moreover, these moves never move the origin.
    Therefore, we can now run the exact same moves on the bottom layer in 2D, until the root component is finished.
    If there is still a high component, it stays connected via the origin.
    We end up with a finished configuration.

    \subparagraph*{Running time.}
    Recall that the \emph{potential} of a cube $c = (c_x, c_y, c_z)$ is $\Pi_c = w_c(c_x + 2c_y + 4c_z)$, where the \emph{weight} $w_c$ depends on the coordinates of $c$ in the following way.
    If $c_z > 1$, then $w_c = 5$, if $c_z = 1$, then $w_c = 4$.
    If $c_z = 0$, then $w_c$ depends on $c_y$.
    If $c_y > 1$, then $w_c = 3$, if $c_y = 1$, then $w_c = 2$ and lastly, if both $c_z = c_y = 0$, then $w_c = 1$.
    The potential of the complete configuration is the sum of potential of the individual cubes.
    Moreover, a sequence of $m$ moves is \emph{safe} if the result is a nonnegative instance $\C'$, such that $\Pi_{\C'} < \Pi_{\C}$ and $m = O(\Pi_{\C} - \Pi_{\C'})$.
    Each move of type \enumit{\ref{enum:pillar_a}--\ref{enum:pillar_j}} strictly reduces the potential function.
    Moreover, each of the moves \enumit{\ref{enum:pillar_a}--\ref{enum:pillar_i}} is safe.
    We show the same for operations of type \enumit{\ref{enum:pillar_j}}:

    \begin{lemma}\label{lem:j_sum_of_coords}
    Operations of type \enumit{\ref{enum:pillar_j}} are safe.
    \end{lemma}
    \begin{proof}
        We will show that any operation of type \enumit{\ref{enum:pillar_j}} strictly decreases $P_\C$ by $O(x)$ and uses $O(x)$ moves to do so.
        First analyse the operations of type \enumit{\ref{enum:pillar_j}} with a pillar of size larger than one, the head~$c = (c_x, c_y, c_z)$ of the pillar involved moves down from its own $z$-coordinate to $z=1$, so $P_\C$ reduces by $4(c_z - 1)$. The cubes beneath $c$ from $z=1$ up to $c_z$ might increase their $x$- or $y$-coordinate by $1$. Therefore, the potential also increases by $2(c_z - 1)$ at most. The cubes that are gathered and then returned do not move positions and therefore do not affect the potential.
        Moreover, $w_c$ for the head becomes one lower, because $c_z$ decreases from $c_z > 1$ to $c_z = 1$.
        In total, the potential~$P_\C$ decreases by $2(c_z - 1) + c_x + c_y$.
        For operations of this type with a pillar of size one, the potential decreases by $c_z + c_x + c_y$, since the head moves from $z=1$ to $z=0$.

        Now we will show that executing one of these operations, which reduces $\Pi_\C$ by $O(c_z + c_x + c_y)$ takes $O(c_z + c_x + c_y)$ moves and is therefore safe.
        Moving via a shortest path over a component with $x$ cubes takes $O(x)$ moves.
        Gathering the seven cubes from $L$ to the clearing pillar and moving them back takes at most $O(c_x + c_y)$ moves, since $L$ cannot reach the origin and has therefore size at most $O(c_x + c_y)$.
        Then, the normal pillar shove takes $O(c_z)$ moves.
        Thus, the total operation takes $O(c_z + c_x + c_y)$ moves and is therefore safe.
    \end{proof}

    Because all operations are safe and reduce the potential, the algorithm performs at most $O(\Pi_\C) = O(X_\C + Y_\C + Z_\C)$ moves.
    For the problem of reconfiguring the cubes into a finished configuration, this is worst-case optimal. 
    An example achieving this bound is a configuration consisting of a path of cubes in a bounding box of equal side lengths $w$ tracing from the origin to the opposite corner of the bounding box. 
    To see that, note that any finished cube at position $(x,y,z)$ requires there to exist $n\geq x y z$ cubes, so at least one of $x$, $y$, and $z$ is at most $n^{1/d}$ for any candidate finished position. 
    There are $\Omega(n)$ cubes $(x',y',z')$ that are initially $\Omega(w-n^{1/3})=\Theta(n)=\Theta(x'+y'+z')$ away from any such potential finished position.

\section{Conclusion}
    We presented an in-place algorithm that reconfigures any configuration of cubes into a compact canonical shape using a number of moves that is proportional to the sum of coordinates of the input cubes. This result is asymptotically optimal. However, just as many other algorithms in the literature, our bounds are amortized in the sense that
    we make use of a number of dedicated cubes which help other cubes move by establishing the necessary connectivity in their neighborhood. This is in particular the case with our pillar shoves, that need some additional cubes to gather at the pillar, to then move up and down the pillar to facilitate moves. These extra moves are charged to one cube in the pillar reducing its coordinates. In the literature such cubes are referred to as \emph{helpers}, \emph{seeds}, or even \emph{musketeers}~\cite{AkitayaADDDFKPP21-musketeers,ConnorM22-musketeers,slideonly-icalp17}.
    
    Such helping cubes are in many ways in conflict with the spirit of modular robot reconfiguration: ideally each module should be able to run the same program more or less independently, without some central control system sending helpers to those places where they are needed. The input-sensitive Gather\&Compact algorithm in 2D by Akitaya et al.~\cite{a2022compacting} does not require amortized analysis and gives a bound on the number of moves for each square in terms of the perimeters of the input and output configurations. 
    The question hence arises whether it is possible to arrive at sum-of-coordinates bounds either in 2D or 3D without amortization?
    For example, is there a compaction algorithm in which each cube in the configuration that starts at position $(x, y, z)$ performs at most $O(x+y+z+a)$ moves, where $a$ is the average $L_1$-distance that cubes lie from the origin?
    
\bibliography{bibliography.bib}
\end{document}